\documentclass[journal,onecolumn]{IEEEtran}
\usepackage{amssymb}
\usepackage{amsmath}
\usepackage{longtable}
\newtheorem{theorem}{Theorem}[section]
\newtheorem{lemma}[theorem]{Lemma}

\newtheorem{corollary}[theorem]{Corollary}

\newenvironment{proof}[1][Proof]{\begin{trivlist}
\item[\hskip \labelsep {\bfseries #1}]}{\end{trivlist}}
\newenvironment{definition}[1][Definition]{\begin{trivlist}
\item[\hskip \labelsep {\bfseries #1}]}{\end{trivlist}}

\makeatletter

\newcommand{\Rmnum}[1]{\expandafter\@slowromancap\romannumeral #1@}
\makeatother

\newcommand{\fp}{{\mathbb F}_{p}}
\newcommand{\fq}{{\mathbb F}_{q}}
\newcommand{\Tr}{{\rm {Tr}}}

\ifCLASSINFOpdf

\else

\fi

\begin{document}
%
\title{New infinite families of
$p$-ary weakly regular bent functions}

\author{ Yanfeng~Qi, Chunming~Tang,
Zhengchun Zhou, Cuiling Fan
\thanks{Y. Qi is with School of Science, Hangzhou Dianzi University, Hangzhou, Zhejiang, 310018, China.
e-mail: qiyanfeng07@163.com.
}
\thanks{C. Tang is with School of Mathematics and Information, China West Normal University, Sichuan Nanchong, 637002, China. e-mail: tangchunmingmath@163.com
}
\thanks{
Z. Zhou is with the School of Mathematics, Southwest Jiaotong University, Chengdu, 610031, China. e-mail:
zzc@swjtu.edu.cn.}
\thanks{
C. Fan is with the School of Mathematics, Southwest Jiaotong University, Chengdu, 610031, China. e-mail:
fcl@swjtu.edu.cn.}
}


\maketitle

\begin{abstract}
The characterization and construction of
bent functions  are challenging problems.
The paper generalizes the constructions of  Boolean bent functions by Mesnager \cite{M2014}, Xu et al. \cite{XCX2015}
and $p$-ary bent functions by
Xu et al. \cite{XC2015} to the construction of $p$-ary weakly regular  bent functions and
presents new infinite families of  $p$-ary weakly regular bent functions from some known weakly regular bent functions (square functions, Kasami functions, and the Maiorana-McFarland class of bent functions). Further, new infinite families of $p$-ary bent idempotents are obtained.
\end{abstract}

\begin{IEEEkeywords}
Bent functions, weakly regular bent functions, bent idempotents, Walsh transform, algebraic degree
\end{IEEEkeywords}

%
\IEEEpeerreviewmaketitle

\section{Introduction}
Boolean bent functions introduced by Rothaus \cite{R1976} in 1976 are an interesting combinatorial object with the maximum
Hamming distance to the set of all affine functions. Such functions have been extensively studied because of
their important applications in cryptograph (stream ciphers \cite{C2010}), sequences  \cite{OSW1982},  and graph theory
\cite{PTF2010}, coding theory ( Reed-Muller codes \cite{CHLL1997}, two-weight and three-weight linear codes   \cite{CK1986,Ding2015}), and association schemes \cite{PTFL2011}. The notation of Boolean
bent functions can be generalized to
functions over an arbitrary finite fields \cite{KSW1985}. Naturally, $p$-ary bent functions are more complicated than Boolean bent functions.
A complete classification of bent functions is still elusive. Further, not only
their characterization, but also
their generation are challenging problems.
Much work on bent functions are devoted to
the construction of bent functions
\cite{CCK2008,C1994,C1996,C2010,CM2011,CG2008,
CK2008,D1974,DD2004,DLCCFG2006,G1968,
L2006,LK2006,LHTK2013,M1973,M2009,M2010,M2011-1,
M2011-2,M2014,MF2013,YG2006}.

Idempotents introduced by Filiol and Fontaine in \cite{F1999,FF1998} are polynomials over
$\mathbb{F}_{2^n}$ such that for any
$x\in \mathbb{F}_{2^n}$,
$f(x)=f(x^2)$. Rotation symmetric Boolean functions, which was also introduced by Filiol and Fontaine under the name of idempotent functions and studied by Pieprzyk and Qu \cite{PQ1999}, are invariant under circular translation of indices. Due to less space to be stored and
allowing faster computation of the Walsh transform, they are of great interest.
They can be obtained from
idempotents (and vice versa) through the choice of a normal basis of $\mathbb{F}_{2^n}$. Characterizing and
constructing idempotent bent functions and rotation symmetric bent functions are difficult and have theoretical and practical interest. In the literature,
few constructions of Boolean bent idempotents have been presented, which are restricted by the degree of
finite fields and have algebraic degree no more than 4. Carlet \cite{C2014}
introduced an open problem on how to construct classes of
bent idempotents over $\mathbb{F}_{2^{2m}}$ of algebraic degree between 5 and $m$.

Recently,  Mesnager \cite{M2014} proposed
several new infinite families of Boolean  bent functions and their duals by adding traces function to known Boolean bent functions. Xu et al. \cite{XCX2015}
generalized Mesnager's work to  more complicated cases by increasing more trace functions. Their work focus on
infinite families of Boolean bent functions. Xu and Cao \cite{XC2015} generalized Mesnager's work to $p$-ary bent functions.
Motivated by their method, we  generalize their work to $p$-ary weakly regular bent functions by adding more trace functions to  known $p$-ary bent functions. Our construction is more general and several new infinite families of $p$-ary weakly regular bent functions are obtained.

Let $p$ be an odd prime and $q=p^n$. Let
$\fq$ denote the finite field with
$q$ elements. This paper considers $p$-ary functions on $\fq$ of the form
$$
f(x)=g(x)+F(\Tr_1^n(u_1x),\Tr_1^n(u_2x)
,\cdots, \Tr_1^n(u_\tau x)),
$$
where $g(x)$ is some known weakly regular  bent function on $\fq$,
$u_1,u_2,\cdots, u_\tau\in \fq$, and
$F(X_1,X_2,\cdots,X_\tau)
\in \mathbb{F}_p[X_1,X_2,\cdots,X_\tau]
$ is a reduced polynomial. From some known bent functions (square functions, Kasami functions, and the Maiorana-McFarland class of bent functions), we present several new infinite
families of $p$-ary weakly regular bent functions. Further we obtain some new infinite families of $p$-ary weakly regular bent idempotents.

The rest of the paper is organized as follows: Section \Rmnum{2} introduces some basic notations and weakly regular bent functions. Section \Rmnum{3} presents several new infinite families of $p$-ary weakly regular bent functions from some known bent functions.
Section \Rmnum{4} makes a conclusion.

\section{Preliminaries}
Throughout this paper, let $p$ be an odd prime, $q=p^n$, and
$\fq$ denote the finite field with
$q$ elements. The trace function from
$\fq$ to a subfield $\mathbb{F}_{p^k}$
is defined by $\Tr_k^n(x)
=\sum_{i=0}^{n/k-1}x^{p^{ik}}$, where
$k|n$. In the case
when $k=1$, $\Tr_1^n(x)$ is called the absolute trace function.

A $p$-ary function is a map from
$\fq$ to $\fp$. The finite  field $\fq$ can be
seen as a $n$-dimensional vector space
$\mathbb{F}_p^n$ over $\fp$. Fix a basis
of the vector space, then
$x\in \fq$ has the unique representation
$(X_1,X_2,\cdots, X_n)$, where
$X_i\in \fp$. And the function $f$ can be represented as a polynomial with the form
$$
f(x)=\sum_{I\in \mathcal{P}(N)}
a_Ix^I,
$$
where $\mathcal{P}(N)$ denotes
the power set of $N=\{1,\cdots,n\}$ and
$a_I\in \fp$. This representation is called the algebraic normal form (ANF) of $f$.
And the degree of such a multivariate polynomial is called the algebraic degree of $f$. A $p$-ary function $f(x)$ on $\fq$ can be seen as a polynomial $F(X_1,X_2,\cdots,
X_n)\in \fp[X_1,X_2,\cdots,
X_n]$. Such $F(X_1,X_2,\cdots,
X_n)$ is called a reduce polynomial.

\begin{definition}
Let $f(x)$ be a $p$-ary function defined on
$\mathbb{F}_{p^n}$. Then $f(x)$ is called an
idempotent if
$$
f(x)=f(x^p), \forall x\in \mathbb{F}_{p^n}.
$$
\end{definition}

\begin{definition}
A $p$-ary function  or
a multivariate polynomial $f(X_1,X_2,\cdots,
X_n)$ is a rotation symmetric polynomial  if it is invariant under cyclic shift:
$$
f(X_n,X_1,X_2,\cdots,
X_{n-1})=f(X_1,X_2,\cdots,
X_n)
$$
\end{definition}

The  Walsh transform of $f$ is defined by
$$
\mathcal{W}_{f}(\beta):=
\sum_{x\in\mathbb{F}_{q}}\zeta_p^{f(x)
+\mathrm{Tr}_{1}^{n}(\beta x)},
$$
where $\beta\in\mathbb{F}_{q}$ and $\zeta_p=e^{2\pi \sqrt{-1}/p}$ is the primitive
$p$-th root of unity.
If we regard $\fq$ as  $\mathbb{F}_{p^{n/2}}\times
\mathbb{F}_{p^{n/2}}$ for even $n$, the Walsh transform of $f$ is
$$
\mathcal{W}_f(\beta_1,\beta_2)=
\sum_{x,y\in \mathbb{F}_{p^{n/2}}}
(-1)^{f(x,y)+\Tr_1^n(\beta_1 x+
\beta_2 y)}.
$$
If we regard $\fq$ as a $n$-dimensional vector space over $\fp$, the Walsh transform of $f$ is
$$
\mathcal{W}_{f}(\beta):=
\sum_{x\in\mathbb{F}_{q}}\zeta_p^{f(x)
+\langle \beta , x\rangle},
$$
where $\langle \beta , x\rangle$ is
the usual dot product on
$\mathbb{F}_p^n$.

\begin{definition}
The function $f(x)$ is a $p$-ary bent functions, if $|\mathcal{W}_f(\beta)|=p^{\frac{n}{2}}$ for any $\beta\in \mathbb{F}_q$.
A  bent function $f(x)$ is regular if there exists some p-ary function $f^*(x)$ satisfying $\mathcal{W}_f(\beta)=p^{\frac{n}{2}}
\zeta_p^{f^*(\beta)}$
for any $\beta \in \mathbb{F}_{q}$.
A  bent function $f(x)$ is weakly regular if
there exists a complex $u$ with unit magnitude
satisfying that
$\mathcal{W}_f(\beta)=up^{\frac{n}{2}}
\zeta_p^{f^*(\beta)}$
for some function $f^*(x)$.
Such function $f^*(x)$ is called the dual of
$f(x)$.
\end{definition}
From \cite{HK2006,HK2010}, a weakly regular bent function
$f(x)$ satisfies that
\begin{equation*}
\mathcal{W}_f(\beta)=\varepsilon \sqrt{p^*}^{n} \zeta_p^{f^*(\beta)},
\end{equation*}
where $\varepsilon =\pm 1$ is called the sign of
the Walsh Transform of $f(x)$ and $p^*={-1 \overwithdelims () p}p$.
The dual of a weakly regular bent function is also
weakly regular bent.
Some results on weakly regular bent functions can be found in \cite{FL2007,HHKWX2009,HK2006,HK2007, HK2010,KSW1985}.

\section{Infinite families of $p$-ary weakly regular bent functions}

Let $g(x)$ be a $p$-ary bent function,
$\tau$ is a positive integer,
$X_1,X_2,\cdots,X_\tau$ be
$\tau$ variables, and
$F(X_1,X_2,\cdots,X_\tau)
\in \mathbb{F}_p[X_1,X_2,\cdots,X_\tau]
$ be a reduced polynomial.

We will study $p$-ary bent functions of the form
\begin{equation}\label{p-bent}
f(x)=g(x)+F(\Tr_1^n(u_1x),\Tr_1^n(u_2x)
,\cdots, \Tr_1^n(u_\tau x)),
\end{equation}
where $u_i\in \mathbb{F}_q$.

As a function from $\fp^{\tau}$ to
the complex field $\mathbb{C}$, $\zeta_p^{F(X_1,
X_2,\cdots,X_\tau)}$ has the unique
Fourier expansion, i.e., there exists
a unique set of
$c_{\mathbf{w}}\in \mathbb{C}$ such that
\begin{equation}\label{fourier}
\zeta_p^{F(X_1,X_2,\cdots,X_\tau)}
=\sum_{\mathbf{w}\in \mathbb{F}_p^{\tau}}
c_{\mathbf{w}}\zeta_p^{w_1X_1+w_2X_2+
\cdots +w_\tau X_\tau},
\end{equation}
where $\mathbf{w}=(w_1,w_2,\cdots,w_\tau)
\in \fp^\tau$.
Equation (\ref{fourier})
holds for any $X_1,X_2,\cdots,X_\tau\in
\mathbb{F}_p$. In particular, take
$X_1=\Tr_1^n(u_1 x)$, $X_2=\Tr_1^n(u_2 x)$,
$\cdots$, $X_\tau=\Tr_1^n(u_\tau x)$. Then for any $x\in \fq$, we have
\begin{equation}\label{fourier1}
\zeta_p^{F(\Tr_1^n(u_1 x),\Tr_1^n(u_2 x),\cdots,\Tr_1^n(u_\tau x))}
=\sum_{\mathbf{w}\in \mathbb{F}_p^{\tau}}
c_{\mathbf{w}}\zeta_p^{\Tr_1^n((\sum_{i=1}^{
\tau}w_iu_i)x)}.
\end{equation}
Multiplying both sides of
Equation (\ref{fourier1}) by
$\zeta_p^{g(x)+\Tr_1^n(\beta x)}$, we have
$$
\zeta_p^{g(x)+F(\Tr_1^n(u_1 x),\Tr_1^n(u_2 x),\cdots,\Tr_1^n(u_\tau x))+\Tr_1^n(\beta x)}
=\sum_{\mathbf{w}\in \mathbb{F}_p^{\tau}}
c_{\mathbf{w}}\zeta_p^{g(x)+\Tr_1^n((\beta+
\sum_{i=1}^{
\tau}w_iu_i)x)}.
$$
Further, we have
$$
\sum_{x\in \fq}\zeta_p^{g(x)+F(\Tr_1^n(u_1 x),\Tr_1^n(u_2 x),\cdots,\Tr_1^n(u_\tau x))+\Tr_1^n(\beta x)}
=\sum_{\mathbf{w}\in \mathbb{F}_p^{\tau}}
c_{\mathbf{w}}\mathcal{W}_g(\beta+\sum_{i=1}^{
\tau}w_iu_i),
$$
i.e.,
$$
\mathcal{W}_f(\beta)
=\sum_{\mathbf{w}\in \mathbb{F}_p^{\tau}}
c_{\mathbf{w}}\mathcal{W}_g(\beta+\sum_{i=1}^{
\tau}w_iu_i).
$$

For further discussion, let
$g(x)$ be a weakly regular bent function
in the rest of the paper, i.e.,
$$
\mathcal{W}_f(\beta)=
\varepsilon\sqrt{p^*}^n\zeta_p^{\widetilde{g}
(\beta)},
$$
where $\widetilde{g}$ is the dual of
$g$, $\varepsilon\in \{1,-1\}$, and
$p^*= {-1 \overwithdelims () p }p
=(-1)^{(p-1)/2}p$.  Hence,
\begin{equation}\label{walshf}
\mathcal{W}_f(\beta)
=\varepsilon \sqrt{p^*}^n \sum_{\mathbf{w}\in \mathbb{F}_p^{\tau}}
c_{\mathbf{w}}\zeta_p^{\widetilde{g}
(\beta+\sum_{i=1}^{
\tau}w_iu_i)}.
\end{equation}

For a general weakly regular bent function,
it is difficult to compute
$\mathcal{W}_f(\beta)$. For some particular
$g(x)$, we can calculate  $\mathcal{W}_f(\beta)$. Suppose that
$\widetilde{g}(x)$ is a quadratic function of the form
\begin{equation}\label{gdual}
\widetilde{g}(x)
=\sum_{i=0}^{n-1}\Tr_1^n(a_ix^{p^i+1})
+\Tr_1^n(bx)+c,
\end{equation}
where $a_i,b\in \fq$ and
$c\in \fp$. For any $1\leq k\leq n-1$,
\begin{align*}
\Tr_1^n(a_k(\beta+
\sum_{i=1}^{\tau}w_iu_i)^{p^k+1})
=&\Tr_1^n(a_k\beta^{p^k+1})
+\sum_{i=1}^\tau w_i\Tr_1^n
(a_k(\beta^{p^k}u_i+\beta u_i^{p^k}))\\&
+\sum_{i=1}^\tau w_i^2\Tr_1^n(a_ku_i^{p^k+1})
+\sum_{1\leq i< j\leq \tau}
w_iw_j\Tr_1^n(a_k(u_i^{p^k}u_j+
u_j^{p^k}u_i)).
\end{align*}
Then
\begin{align*}
\sum_{k=0}^{n-1}\Tr_1^n(a_k(\beta+
\sum_{i=1}^{\tau}w_iu_i)^{p^k+1})
=&\sum_{k=0}^{n-1}\Tr_1^n(a_k\beta^{p^k+1})
+\sum_{i=1}^\tau w_i(\Tr_1^n
(u_i\sum_{s=0}^{n-1}a_k \beta^{p^k})+
\Tr_1^n(\beta \sum_{k=0}^{n-1}a_ku_i^{p^k}))\\&
+\sum_{i=1}^\tau w_i^2\Tr_1^n(
\sum_{k=0}^{n-1}a_ku_i^{p^k+1})
+\sum_{1\leq i< j\leq \tau}
w_iw_j\Tr_1^n(
\sum_{k=0}^{n-1}(a_k(u_i^{p^k}u_j+
u_i u_j^{p^k}))).
\end{align*}
From $\Tr_1^n(b(\beta+
\sum_{i=1}^{\tau}w_iu_i))
+c=
\Tr_1^n(b\beta)+c+\sum_{i=1}^{\tau}
w_i\Tr_1^n(bu_i)$,
\begin{align*}
\widetilde{g}(\beta
+\sum_{i=1}^{\tau}w_iu_i)
=&\widetilde{g}(\beta)+\sum_{i=1}^\tau w_i(\sum_{s=0}^{n-1}\Tr_1^n
(a_k (u_i\beta^{p^k}+\beta u_i^{p^k}))+
\Tr_1^n(bu_i))\\&
+\sum_{i=1}^\tau w_i^2\sum_{k=0}^{n-1}\Tr_1^n(
a_ku_i^{p^k+1})
+\sum_{1\leq i< j\leq \tau}
w_iw_j(\sum_{k=0}^{n-1}\Tr_1^n
(a_k(u_i^{p^k}u_j+
u_i u_j^{p^k}))).
\end{align*}
Then, we have the following lemma.
\begin{lemma}\label{lem}
Let $g(x)$ be a weakly regular bent function, $\mathcal{W}_g(\beta)
=\varepsilon \sqrt{p^*}^n\zeta_p^{
\widetilde{g}(\beta)}$, and the dual
$\widetilde{g}$ be
$$
\widetilde{g}(x)
=\sum_{k=0}^{n-1}\Tr_1^n(a_kx^{p^k+1})
+\Tr_1^n(bx)+c,
$$
where $a_k,b\in \mathbb{F}_q$ and
$c\in \mathbb{F}_p$.
Let $u_1,u_2,\cdots,u_\tau\in \fq$ such that $$
\sum_{k=0}^{n-1}\Tr_1^n(a_k(u_i^{p^k}u_j
+u_iu_j^{p^k}))=0.
$$
Then the $p$-ary function $f(x)$ defined in
Equation (\ref{p-bent}) is weakly regular bent. Further, the Walsh transform of
$f$ is
$$
\mathcal{W}_f(\beta)
=\varepsilon \sqrt{p^*}^n
\zeta_p^{\widetilde{g}
(\beta)+F(X_1,X_2,\cdots,X_\tau)},
$$
where $X_i=
\sum_{k=0}^{n-1}\Tr_1^n(a_k(u_i\beta^{p^k}
+u_i^{p^k}\beta))+\Tr_1^n(bu_i)$.
\end{lemma}
\begin{proof}
The Walsh transform of $f$ is
$$
\mathcal{W}_f(\beta)
=\varepsilon \sqrt{p^*}^n
\zeta_p^{\widetilde{g}(\beta)}\sum_{\mathbf{w}\in \mathbb{F}_p^{\tau}}
c_{\mathbf{w}}\zeta_p^{\sum_{i=1}^{
\tau}w_iX_i},
$$
where $X_i=
\sum_{k=0}^{n-1}\Tr_1^n(a_k(u_i\beta^{p^k}
+u_i^{p^k}\beta))+\Tr_1^n(bu_i)$.
From Equation (\ref{fourier}),
we have
$$
\mathcal{W}_f(\beta)
=\varepsilon \sqrt{p^*}^n
\zeta_p^{\widetilde{g}
(\beta)+F(X_1,X_2,\cdots,X_\tau)}.
$$
Hence, this lemma follows.
\end{proof}
From Lemma \ref{lem} and some known weakly regular bent functions, we construct three new infinite families of $p$-ary weakly regular bent functions.
\subsection{New infinite family of $p$-ary
weakly regular bent functions from square functions}

Let the square function $g(x)=\Tr_1^n(\lambda x^2)$ be weakly regular bent, where
$\lambda\in \fq^{\times}$.
From \cite{HK2006}, the Walsh transform of $g(x)$ is
$$
\mathcal{W}_g(\beta)
=(-1)^{n-1}\eta(\lambda)
\sqrt{p^*}^n\zeta_p^{-\Tr_1^n({\beta^2 \over 4\lambda})},
$$
where $\beta \in \fq$ and $\eta$ is the multiplicative quadratic character of
$\fq$.  The dual of $g$ is $\widetilde{g}(x)
=-\Tr_1^n({x^2 \over 4\lambda})$.

Let $\mathcal{N}$ be a linear subspace of
$\fq$ over $\fp$.
The subspace $\mathcal{N}$ on
$\widetilde{g}$ is a self-orthogonal subspace if for any $x,y\in \mathcal{N}$,
$\sum_{k=0}^{n-1}\Tr_1^n(a_k(x^{p^k}y
+xy^{p^k}))=0$, where
$a_k$ are determined by $\widetilde{g}$.
The subspace
$\mathcal{N}$ is called a maximal
self-orthogonal subspace if
for any self-orthogonal subspace
$\mathcal{N}'$ containing $\mathcal{N}$,
then $\mathcal{N}'=\mathcal{N}$.

\begin{lemma}\label{lem-A}
Let $g(x)$ be a weakly regular bent function, $\mathcal{W}_g(\beta)
=\varepsilon \sqrt{p^*}^n\zeta_p^{
\widetilde{g}(\beta)}$, and the dual
$\widetilde{g}$ be
$$
\widetilde{g}(x)
=\sum_{k=0}^{n-1}\Tr_1^n(a_kx^{p^k+1})
+\Tr_1^n(bx)+c,
$$
where $a_k,b\in \mathbb{F}_q$ and
$c\in \mathbb{F}_p$.
Let $\mathcal{N}$ on $\widetilde{g}$ be a maximal self-orthogonal subspace,
$\tau=dim(\mathcal{N})$,
$u_1,\cdots,u_\tau$ be a basis of
$\mathcal{N}$ over $\fp$, and  $F(X_1,X_2,\cdots,X_\tau)$ be a reduced
polynomial in $\fp[X_1,X_2,\cdots,
X_\tau]$ of algebraic degree $d$.
Then the $p$-ary function $f(x)$ defined in
Equation (\ref{p-bent}) is a weakly regular bent function of algebraic degree $d$.
\end{lemma}
\begin{proof}
From the definition of maximal
self-orthogonal subspaces and Lemma \ref{lem}, this lemma follows.
\end{proof}
\begin{theorem}\label{thm-A}
Let $\tau$ be a positive integer,
$u_1,\cdots,u_\tau\in \fq^{\times}$,
and  $F(X_1,X_2,\cdots,X_\tau)$ be a reduced
polynomial in $\fp[X_1,X_2,\cdots,
X_\tau]$, where
$\Tr_1^n({1\over \lambda}u_iu_j)
=0$ for any $1\leq i\leq j\leq \tau$.
Then, the $p$-ary function
$$
f(x)=\Tr_1^n(\lambda x^2)
+F(\Tr_1^n(u_1 x),\Tr_1^n(u_2 x), \cdots,
\Tr_1^n(u_\tau x))
$$
is a weakly regular bent function.
\end{theorem}
\begin{proof}
In Lemma \ref{lem-A}, take
$g(x)=\Tr_1^n(\lambda x^2)$. Then
this theorem follows.
\end{proof}
\begin{corollary}\label{cor-A1}
Let $n=pm$, $u_1,u_2,\cdots, u_m$ be a basis of $\mathbb{F}_{p^m}$ over
$\fp$, $\lambda \in \mathbb{F}_{p^m}^{\times}$, and $F(X_1,X_2,\cdots,X_m)$ be a reduced
polynomial in $\fp[X_1,X_2,\cdots,
X_m]$. Then the $p$-ary function
$$
f(x)=\Tr_1^n(\lambda x^2)
+F(\Tr_1^n(u_1 x),\Tr_1^n(u_2 x), \cdots,
\Tr_1^n(u_m x))
$$
is a weakly regular bent function.
\end{corollary}
\begin{proof}
For any $i,j$, $\Tr_1^n(
{1\over \lambda}u_iu_j)
=\Tr_1^m(\Tr_m^n({1\over \lambda}u_iu_j))
=\Tr_1^m(p{1\over \lambda}u_iu_j)
=0$.
From Theorem \ref{thm-A}, this corollary follows.
\end{proof}
In Corollary \ref{cor-A1}, take
the normal basis of $\mathbb{F}_{p^m}$.
And we have the following corollary.
\begin{corollary}\label{cor-A2}
Let $n=pm$, $u$ be a normal element of $\mathbb{F}_{p^m}$, $\lambda \in \mathbb{F}_{p^m}^{\times}$, and $F(X_1,X_2,\cdots,X_m)$ be a reduced
polynomial in $\fp[X_1,X_2,\cdots,
X_m]$. Then the $p$-ary function
$$
f(x)=\Tr_1^n(\lambda x^2)
+F(\Tr_1^n(u x),\Tr_1^n(u^p x), \cdots,
\Tr_1^n(u^{p^{m-1}} x))
$$
is a weakly regular bent function.
In particular, when $\lambda
\in \fp^{\times}$ and
$F(X_1,X_2,\cdots,X_m)$ is a rotation symmetric polynomial, $f(x)$ is a
$p$-ary bent idempotent.
\end{corollary}
\begin{proof}
From Corollary \ref{cor-A1},
$f(x)$ is weakly regular bent.
When $\lambda
\in \fp^{\times}$ and
$F(X_1,X_2,\cdots,X_m)$ is a rotation symmetric polynomial, we have
\begin{align*}
f(x^p)=& \Tr_1^n(\lambda x^{2p})
+F(\Tr_1^n(u x^p),\Tr_1^n(u^p x^p), \cdots,
\Tr_1^n(u^{p^{m-1}} x^p))\\
=&\Tr_1^n(\lambda x^2)
+F(\Tr_1^n(u^{p^{m-1}} x),\Tr_1^n(u^{p^0} x), \cdots,
\Tr_1^n(u^{p^{m-2}} x))\\
=& \Tr_1^n(\lambda x^2)
+F(\Tr_1^n(u x),\Tr_1^n(u^p x), \cdots,
\Tr_1^n(u^{p^{m-1}} x))\\
=&f(x).
\end{align*}
Hence, $f(x)$ is a $p$-ary bent idempotent.
\end{proof}

\begin{corollary}\label{cor-A3}
Let $u\in \fq^{\times}$,
$\Tr_1^n({1\over \lambda}u^2)=0$,  and $F(X)$ be a reduced
polynomial in $\fp[X]$. Then the $p$-ary function
$$
f(x)=\Tr_1^n(\lambda x^2)
+F(\Tr_1^n(u x))
$$
is a weakly regular bent function.
\end{corollary}
\begin{proof}
From Theorem \ref{thm-A}, this corollary
follows.
\end{proof}

\subsection{New infinite family of $p$-ary
weakly regular bent functions from Kasami functions}

Let $n=2k$ and $g(x)=\Tr_1^n(ax^{p^k+1})$
be
the $p$-ary Kasami function.
Liu and Komo \cite{LK1992} proved that
$g(x)$ is  bent. Helleseth et al. \cite{HK2006}
showed that the Walsh transform of
$g(x)$ is
$$
\mathcal{W}_g(\beta)=-p^k\zeta_p^{
-\Tr_1^k({\beta^{p^k+1}\over a+a^{p^k}})}.
$$
And the dual of $g$ is
$\widetilde{g}(x)=-\Tr_1^k({x^{p^k+1}\over a+a^{p^k}})$.

\begin{theorem}\label{thm-B}
Let $n=2k$, $\tau$ be a positive integer,
$u_1,u_2,\cdots,u_\tau\in \fq^{\times}$,
and $F(X_1,X_2,\cdots,X_\tau)$ be a reduced
polynomial in $\fp[X_1,X_2,\cdots,
X_\tau]$, where
$\Tr_1^k({1\over a+a^{p^k}}(u_i^{p^k}
u_j+u_iu_j^{p^k}))
=0$ for any $1\leq i\leq j\leq \tau$.
Then the $p$-ary function
$$
f(x)=\Tr_1^n(a x^{p^k+1})
+F(\Tr_1^n(u_1 x),\Tr_1^n(u_2 x), \cdots,
\Tr_1^n(u_\tau x))
$$
is a weakly regular bent function.
\end{theorem}
\begin{proof}
The trace function $\Tr_k^n$ from
$\mathbb{F}_{p^n}$ to
$\mathbb{F}_{p^k}$ is surjective.
There exists $A\in \mathbb{F}_{p^n}$ such that $\Tr_k^n(A)={1\over a+a^{p^k}}
=A^{p^k}+A$. The dual $\widetilde{g}$ of
$g$ is $\widetilde{g}(x)
=-\Tr_1^k({1\over a+a^{p^k}}x^{p^k+1})
=-\Tr_1^n(Ax^{p^k+1})$. Then
\begin{align*}
Tr_1^n(A(u_i^{p^k}+u_iu_j^{p^k}))
=&  Tr_1^k((A+A^k)(u_i^{p^k}u_j+u_iu_j^{p^k})\\
=& Tr_1^k({1\over a+a^k} (u_i^{p^k}u_j+
u_iu_j^{p^k})\\
=& 0.
\end{align*}
From Lemma \ref{lem}, this theorem follows.
\end{proof}
\begin{corollary}\label{cor-B1}
Let $n=2k$, $\tau$ be a positive integer,
$\lambda,u_1,u_2,\cdots,u_\tau\in \mathbb{F}_{p^k}^{\times}$,
and $F(X_1,X_2,\cdots,X_\tau)$ be a reduced
polynomial in $\fp[X_1,X_2,\cdots,
X_\tau]$, where
$\Tr_1^k({1\over \lambda}u_i
u_j)
=0$ for any $1\leq i\leq j\leq \tau$.
Then the $p$-ary function
$$
f(x)=\Tr_1^k(\lambda x^{p^k+1})
+F(\Tr_1^n(u_1 x),\Tr_1^n(u_2 x), \cdots,
\Tr_1^n(u_\tau x))
$$
is a weakly regular bent function.
\end{corollary}
\begin{proof}
Take $a\in \mathbb{F}_{p^n}$ such that
$\Tr_k^n(a)=\lambda$. Then
$\Tr_1^n(ax^{p^k+1})=\Tr_1^k(\lambda x^{p^k+1})$ and
\begin{align*}
\Tr_1^k({1\over \lambda} (u_i^{p^k}u_j
+u_iu_j^{p^k})=\Tr_1^k({2\over \lambda}
u_iu_j)=2\Tr_1^k({1\over \lambda} u_iu_j)
=0.
\end{align*}
From Theorem \ref{thm-B}, this corollary follows.
\end{proof}

\begin{corollary}\label{cor-B2}
Let $n=2k$, $\tau$ be a positive integer,
$\lambda\in \mathbb{F}_{p^k}^{\times}$,
$u\in \fq^{\times}$,
and $F(X)$ be a reduced
polynomial in $\fp[X]$, where
$\Tr_1^k({1\over \lambda}u^{p^k+1})
=0$.
Then the $p$-ary function
$$
f(x)=\Tr_1^k(\lambda x^{p^k+1})
+F(\Tr_1^n(u x))
$$
is a weakly regular bent function.
\end{corollary}
\begin{proof}
From Theorem \ref{thm-B}, take
$a\in \fq$ such that $\Tr_k^n(a)=\lambda$.
Then this corollary follows.
\end{proof}

\begin{corollary}\label{cor-B3}
Let $n=2k=2pm$,
$\lambda\in \mathbb{F}_{p^m}^{\times}$, $u$  be a normal element of  $\mathbb{F}_{p^m}^{\times}$,
and $F(X_1,X_2,\cdots,X_m)$ be a reduced
polynomial in $\fp[X_1,X_2,\cdots,
X_m]$.
Then the $p$-ary function
$$
f(x)=\Tr_1^k(\lambda x^{p^k+1})
+F(\Tr_1^n(u x),\Tr_1^n(u^p x), \cdots,
\Tr_1^n(u^{p^{m-1}} x))
$$
is a weakly regular bent function.
Further, when $\lambda\in \fp^{\times}$ and
$F(X_1,X_2,\cdots,X_m)$ is rotation symmetric, then $f(x)$ is a bent idempotent.
\end{corollary}
\begin{proof}
Note that $\Tr_1^k({1\over \lambda}
u^{p^i}u^{p^j})=\Tr_1^m(p{1\over \lambda}
u^{p^i}u^{p^j})=0$. From Corollary
\ref{cor-B1}, $f(x)$ is weakly regular bent. When $\lambda\in \fp^{\times}$ and
$F(X_1,X_2,\cdots,X_m)$ is rotation symmetric, then $f(x^p)=f(x)$. Hence,
$f(x)$ is a $p$-ary bent idempotent.
\end{proof}

\subsection{New infinite family of $p$-ary
weakly regular bent functions from the
Maiorana-McFarland class}

In this subsection, we identify
$\mathbb{F}_{p^n}~(n=2k)$ as
$\mathbb{F}_{p^k}\times \mathbb{F}_{p^k}$ and consider $p$-ary functions with bivariate representation
$f(x,y)=\Tr_1^k(P(x,y))$, where
$P(x,y)$ is a polynomial in two-variable
over $\mathbb{F}_{p^k}$.
For $(\beta_1,\beta_2),
(x,y)\in \mathbb{F}_{p^k}\times \mathbb{F}_{p^k}$, the scalar product in $\mathbb{F}_{p^k}\times \mathbb{F}_{p^k}$ can be defined as
$$
\langle (\beta_1,\beta_2),
(x,y) \rangle =\Tr_1^k(\beta_1 x+\beta_2 y).
$$
The well-known Maiorana-McFarland class of
$p$-ary bent functions can be defined as follows:
\begin{equation}\label{MM-g}
g(x,y)=\Tr_1^k(x\pi(y))+h(y), ~
(x,y)\in \mathbb{F}_{p^k}\times \mathbb{F}_{p^k},
\end{equation}
where $\pi: \mathbb{F}_{p^k}
\rightarrow \mathbb{F}_{p^k}$ is a
permutation and $h$ is any
$p$-ary function over $\mathbb{F}_{p^k}$.
From \cite{C2010}, the dual $\widetilde{g}$ is
\begin{equation}\label{MM-dual}
\widetilde{g}(x,y)=\Tr_1^k(y\pi^{-1}(x))
+h(\pi^{-1}(x)),
\end{equation}
where $\pi^{-1}$ denotes the inverse
mapping of $\pi$.

By choosing suitable permutation
$\pi$, we will construct a new infinite family of $p$-ary weakly regular bent functions.
\begin{theorem}\label{thm-C}
Let $n=2k$, $\pi$ be a linearized  permutation polynomial over $\mathbb{F}_{p^k}$, and
$u_1=(u_1^{(1)},u_1^{(2)}),
\cdots, u_\tau=(u_\tau^{(1)},u_\tau^{(2)})
\in \mathbb{F}_{p^k}\times
\mathbb{F}_{p^k}$ such that
$$
\Tr_1^k(u_i^{(2)}\pi^{-1}(u_{j}^{(1)})
+u_j^{(2)}\pi^{-1}(u_{i}^{(1)}))=0~, 1\leq i\leq j\leq \tau.
$$
Then the $p$-ary function
$$
f(x,y)=\Tr_1^k(x\pi(y))+
\Tr_1^k(by)+F(\Tr_1^k(u_1^{(1)}x+
u_1^{(2)}y),\cdots, \Tr_1^k(u_\tau^{(1)}x+
u_\tau^{(2)}y))
$$
is weakly regular bent.
\end{theorem}
\begin{proof}
Take $h(y)=\Tr_1^k(by)$ and
$g(x)=\Tr_1^k(x\pi(y))+\Tr_1^k(by)$. Then
from Equation (\ref{MM-dual}),
\begin{align*}
\widetilde{g}(\beta_1+\sum_{i=1}^{\tau}
w_iu_i^{(1)}, \beta_2 + \sum_{i=1}^{\tau}
w_iu_i^{(2)})=&\widetilde{g}(\beta_1,\beta_2)
+\sum_{i=1}^{\tau} w_iX_i\\
&+\sum_{i=1}^{\tau} w_i^2 \Tr_1^k(u_i^{(2)}
\pi^{-1}u_i^{(1)})+
\sum_{1\leq i< j\leq \tau}w_iw_j\Tr_1^k
(u_i^{(2)}\pi^{-1}(u_{j}^{(1)})
+u_j^{(2)}\pi^{-1}(u_{i}^{(1)})),
\end{align*}
where $X_i=
\Tr_1^k((\beta_2+b)\pi^{-1}(u_{i}^{(1)})
+u_i^{(2)}\pi^{-1}(\beta_1))$.
Since for any $1\leq i\leq j\leq \tau$
$\Tr_1^k(u_i^{(2)}\pi^{-1}(u_{j}^{(1)})
+u_j^{(2)}\pi^{-1}(u_{i}^{(1)}))=0$, from
Equation (\ref{walshf}),
$$
\mathcal{W}_f(\beta_1,\beta_2)
=\varepsilon \sqrt{p^*}^{2k}
\zeta_p^{\widetilde{g}(\beta_1,\beta_2)} \sum_{\mathbf{w}\in \mathbb{F}_p^{\tau}}
c_{\mathbf{w}}\zeta_p^{\sum_{i=1}^{
\tau}w_iX_i}.
$$
From Equation (\ref{fourier}),
$$
\mathcal{W}_f(\beta_1,\beta_2)
=\varepsilon \sqrt{p^*}^{2k}
\zeta_p^{\widetilde{g}(\beta_1,\beta_2)
+F(X_1,\cdots,X_\tau)}.
$$
Hence, $f(x,y)$ is a $p$-ary weakly regular bent function.
\end{proof}

\section{Conclusion}
In this paper, we generalize the work of Mesnager \cite{M2014} and   Xu et al. \cite{XCX2015,XC2015} to $p$-ary weakly regular bent functions. From known weakly regualr bent functions (square functions, Kasami functions, and the Maiorana-McFarland class of bent functions), we construct three new infinite families of $p$-ary weakly regular bent functions, which  contain
some infinite families of $p$-ary
bent idempotents.

\section*{Acknowledgment}
This work was supported by
the National Natural Science Foundation of China
(Grant No. 11401480, No.10990011 \& No. 61272499).
Yanfeng Qi also acknowledges support from
KSY075614050 of Hangzhou Dianzi University.


\ifCLASSOPTIONcaptionsoff
  \newpage
\fi


\begin{thebibliography}{99}

\bibitem{CK1986} R. Calderbank and W. M. Kantor, ``The geometry of two-weight codes,"
Bull. London Math. Soc., vol. 18, no. 2, pp. 97-122, 1986.

\bibitem{CCK2008} A. Canteaut, P. Charpin, and G. Kyureghyan, ``A new class of monomial
bent functions," Finite Fields Their Appl., vol. 14, no. 1, pp. 221-241,
2008.
\bibitem{C1994} C. Carlet, ``Two new classes of bent functions," in EUROCRYPT
(Lecture Notes in Computer Science), vol. 765. New York, NY, USA:
Springer-Verlag, 1994, pp. 77-101.
\bibitem{C1996} C. Carlet, ``A construction of bent function," in Proc. 3rd Int. Conf.
Finite Fields and Appl., 1996, pp. 47-58.

\bibitem{C2010} C. Carlet, ``Boolean functions for cryptography and error correcting
codes," in Boolean Models and Methods in Mathematics, Computer
Science, and Engineering, Y. Crama and P. L. Hammer, Eds. Cambridge,
U.K.: Cambridge Univ. Press, 2010, pp. 257-397.
\bibitem{C2014} C Carlet. ``Open Open Problems on Binary Bent Functions". Open Problems in Mathematics and Computational Science 2014, pp 203-241

\bibitem{CM2011} C. Carlet and S. Mesnager, ``On Dillon¡¯s class H of bent functions,
Niho bent functions and O-polynomials," J. Combinat. Theory, Ser. A,
vol. 118, no. 8, pp. 2392-2410, 2011.

\bibitem{CG2008} P. Charpin and G. Gong, ``Hyperbent functions, Kloosterman sums and
Dickson polynomials," in Proc. ISIT, Jul. 2008, pp. 1758-1762.
\bibitem{CK2008} P. Charpin and G. Kyureghyan, ``Cubic monomial bent functions:
A subclass of M," SIAM J. Discrete Math., vol. 22, no. 2, pp. 650-665,
2008.
\bibitem{CHLL1997} G. Cohen, I. Honkala, S. Litsyn, and A. Lobstein, Covering Codes.
Amsterdam, The Netherlands: North Holland, 1997.
\bibitem{D1974} J. Dillon, ``Elementary Hadamard difference sets," Ph.D. dissertation,
Netw. Commun. Lab., Univ. Maryland, College Park, MD, USA, 1974.


\bibitem{DD2004} J. F. Dillon and H. Dobbertin, ``New cyclic difference sets with Singer
parameters," Finite Fields Their Appl., vol. 10, no. 3, pp. 342-389, 2004.
\bibitem{Ding2015} C. Ding, ``Linear codes from some 2-designs," IEEE Trans. Inform. Theory, vol. 61, no. 6, pp. 3265-3275,  June 2015.
\bibitem{DLCCFG2006} H. Dobbertin, G. Leander, A. Canteaut, C. Carlet, P. Felke, and
P. Gaborit, ``Construction of bent functions via Niho power functions,"
J. Combinat. Theory, Ser. A, vol. 113, no. 5, pp. 779-798, 2006.
\bibitem{FL2007}  K. Feng, J. Luo, ``Value distributions of exponential sums from perfect nonlinear functions and their applications," IEEE Trans. Inform. Theory,  vol. 53,  no. 9, pp. 3035-3041, August 2007.

\bibitem{F1999} C. Fontaine, ``On some cosets of the first-order Reed-Muller code with high minimum weight," IEEE Trans. Inform. Theory 45, 1237-1243 (1999)
\bibitem{FF1998} E. Filiol, C. Fontaine, ``Highly nonlinear balanced Boolean functions with a good correlationimmunity," in Proceedings of EUROCRYPT¡¯98. Lecture Notes in Computer Science, vol. 1403 (1998), pp. 475-488

\bibitem{G1968} R. Gold, ``Maximal recursive sequences with 3-valued recursive crosscorrelation functions (Corresp.),"  IEEE Trans. Inf. Theory, vol. 14, no. 1,
pp. 154-156, Jan. 1968.

\bibitem{HHKWX2009} T. Helleseth, H. D. L. Hollmann, A. Kholosha, Z. Wang, and Q. Xiang,
``Proofs of two conjectures on ternary weakly regular bent functions,"
IEEE Trans. Inf. Theory, vol. 55, no. 11, pp. 5272-5283, Nov. 2009.

\bibitem{HK2006} T. Helleseth and A. Kholosha, ``Monomial and quadratic bent functions over the finite fields of odd characteristic," IEEE Trans. Inf. Theory, vol. 52, no.5, pp. 2018-2032, May 2006.

\bibitem{HK2007} T. Helleseth and A. Kholosha, ``On the dual of monomial quadratic
p-ary bent functions," in Sequences, Subsequences, and Consequences,
ser. Lecture Notes in Computer Science, S. Golomb, G. Gong, T. Helleseth, and H. Y. Song, Eds. Berlin: Springer-Verlag, 2007, vol. 4893, pp. 50-61.
\bibitem{HK2010} T. Helleseth and A. Kholosha, ``New binomial bent functions over the
finite fields of odd characteristic," IEEE Trans. Inf. Theory, vol. 56, no.
9, pp. 4646-4652, Sep. 2010.

\bibitem{KSW1985}  P.V. Kumar, R.A. Scholtz, L.R. Welch, ``Generalized bent functions and their properties", J. Combin. Theory Ser. A 40 (1985) 90-107.

\bibitem{L2006} G. Leander, ``Monomial bent functions," IEEE Trans. Inf. Theory,
vol. 52, no. 2, pp. 738-743, Feb. 2006.
\bibitem{LK2006} G. Leander and A. Kholosha, ``Bent functions with 2r Niho exponents,"
IEEE Trans. Inf. Theory, vol. 52, no. 12, pp. 5529-5532, Dec. 2006.
\bibitem{LHTK2013} N. Li, T. Helleseth, X. Tang, and A. Kholosha, ``Several new classes of
bent functions from Dillon exponents," IEEE Trans. Inf. Theory, vol. 59,
no. 3, pp. 1818-1831, Mar. 2013.

\bibitem{LK1992}S.C. Liu, J.J. Komo, ``Nonbinary Kasami sequences over GF(p)," IEEE Trans. Inf. Theory 38(4), 1409-1412 (1992)



\bibitem{M1973} R. L. McFarland, ``A family of noncyclic difference sets," J. Combinat.
Theory, Ser. A, vol. 15, no. 1, pp. 1-10, 1973.

\bibitem{M2009} S. Mesnager, ``A new family of hyper-bent Boolean functions in polynomial form," in Cryptography and Coding (Lecture Notes in Computer
Science), vol. 5921, M. G. Parker, Ed. Berlin, Germany: Springer-Verlag,
2009, pp. 402-417.
\bibitem{M2010} S. Mesnager, ``Hyper-bent Boolean functions with multiple trace terms,"
in Arithmetic of Finite Fields (Lecture Notes in Computer Science),
vol. 6087, M. Hasan and T. Helleseth, Eds. Berlin, Germany:
Springer-Verlag, 2010, pp. 97-113.
\bibitem{M2011-1} S. Mesnager, ``Bent and hyper-bent functions in polynomial form and
their link with some exponential sums and Dickson polynomials,¡± IEEE
Trans. Inf. Theory, vol. 57, no. 9, pp. 5996-6009, Sep. 2011.
\bibitem{M2011-2} S. Mesnager, ``A new class of bent and hyper-bent Boolean functions
in polynomial forms," Des., Codes Cryptography, vol. 59, nos. 1-3,
pp. 265-279, 2011.
\bibitem{M2014} S. Mesnager, ``Several New Infinite Families of Bent Functions and Their Duals," IEEE Trans. Inf. Theory, vol. 60, no. 7, JULY 2014
\bibitem{MF2013} S. Mesnager and J. P. Flori, ``Hyper-bent functions via Dillon-like
exponents," IEEE Trans. Inf. Theory, vol. 59, no. 5, pp. 3215-3232,
May 2013.

\bibitem{OSW1982} J. D. Olsen, R. A. Scholtz, and L. R. Welch, ``Bent-function sequences,"
IEEE Trans. Inf. Theory, vol. 28, no. 6, pp. 858-864, Nov. 1982.
\bibitem{PQ1999} J. Pieprzyk, C. Qu, ``Fast Hashing and rotation symmetric functions," J. Univ. Comput. Sci. 5, 20-31 (1999)
\bibitem{PTF2010} A. Pott, Y. Tan, and T. Feng, ``Strongly regular graphs associated with
ternary bent functions," J. Combinat. Theory, Ser. A, vol. 117, no. 6,
pp. 668-682, 2010.
\bibitem{PTFL2011} A. Pott, Y. Tan, T. Feng, and S. Ling, ``Association schemes arising
from bent functions," Des., Codes Cryptography, vol. 59, nos. 1-3,
pp. 319-331, 2011.
\bibitem{R1976} O. Rothaus, ``On bent functions," J. Combinat. Theory, Ser. A, vol. 20, no. 3, pp. 300-305, 1976.

\bibitem{XCX2015} G. Xu, X. Cao,  S. Xu. ``Several new classes of Boolean functions with few Walsh transform values", arXiv:1506.04886v1
\bibitem{XC2015}  G. Xu, X. Cao, ``Several classes of bent, near-bent and 2-plateaued functions over finite fields of odd characteristic," arXiv:1508.03415

\bibitem{YG2006} N. Y. Yu and G. Gong, ``Construction of quadratic bent functions
in polynomial forms," IEEE Trans. Inf. Theory, vol. 52, no. 7,
pp. 3291-3299, Jul. 2006.

\end{thebibliography}
\end{document}